\documentclass[10pt]{article}

\usepackage{amssymb}
\usepackage{amsfonts}
\usepackage{amsmath}
\usepackage{amscd}
\usepackage[all,cmtip]{xy}
\usepackage{amsthm}
\usepackage{setspace}
\usepackage{enumerate}
\usepackage{graphicx}
\usepackage{float}
\usepackage{tikz}

\theoremstyle{plain}

\newtheorem{thm}{Theorem}

\theoremstyle{definition}

\newcommand{\CC}{{\mathbb C}}

\newcommand{\ZZ}{{\mathbb Z}}
\newcommand{\PP}{{\mathbb P}}
\newcommand{\Spec}{{\rm Spec}}

\newcommand{\Hom}{{\rm Hom}}
\newcommand{\GL}{{\rm GL}}
\newcommand{\D}{{\rm D}}

\allowdisplaybreaks

\title{Fourier Duality of Quantum Curves}
\author{}
\author{Martin T. Luu, Albert  Schwarz}
\date{}

\begin{document}

\newcommand{\Addresses}{{
\bigskip
\footnotesize

M. Luu, \textsc{Department of Mathematics, University of Illinois at Urbana-Champaign, IL 61801, USA} \par \nopagebreak \textit{E-mail address:} \texttt{mluu@illinois.edu}

\medskip

A. Schwarz, \textsc{Department of Mathematics, University of California, Davis, CA 95616, USA} \par \nopagebreak \textit{E-mail address:} \texttt{schwarz@math.ucdavis.edu} }}

\maketitle

\begin{abstract}
There are two different ways to deform a quantum curve along the flows of the KP hierarchy. We clarify the relation between the two KP orbits: In the framework of suitable connections attached to the quantum curve they are related by a local Fourier duality. As an application we give a conceptual proof of duality results in 2D quantum gravity. \end{abstract}

\section{Introduction}
One way to define a quantum curve, see for example \cite{SCH2}, is as a pair $(P,Q)$ of ordinary differential operators in $\CC[\![x]\!][\partial_{x}]$ such that
$$[P,Q]=\hbar$$ 
Here $\hbar$ might be viewed as a formal parameter or a fixed complex number depending on the situation of interest. For the remainder of this work we set $\hbar = 1$. We say that the quantum curve has bi-degree  $(p,q)$ if $P$ is a differential operator of order $p$ and $Q$ is a differential operator of order $q.$ We will work with scalar differential operators, however, our methods can be applied also to matrix differential operators. The notion of quantum curve originates from the fact that given a complex algebraic curve and two suitable functions $f_{1},f_{2}$ on it, it is known by work of Burchnall and Chaundy \cite{BC} and  Krichever \cite {KRI}, that one can construct two commuting scalar differential operators $P_{0}$ and $Q_{0}$. Hence an algebraic curve is related to the  classical situation $\hbar = 0$. 

Notice that there are two natural ways to deform a quantum curve along the flows of the Kadomtsev-Petviashvili (KP) hierarchy of partial differential equations. Namely, if $P$  is a normalized differential operator in the sense that is has the form
$\partial_{x}^{p}+a_{p-2}\partial_{x}^{p-2}+\cdots$, 
then we can define a pseudodifferential operator $L=P^{\frac{1}{p}}$. In terms of this operator we can define a family of quantum curves $(P(t_{1},\cdots,t_{p+q}),Q(t_{1},\cdots,t_{p+q}))$ of bi-degree $(p,q)$ solving the KP-equations
$$\frac{\partial }{\partial t_n} L(t_{1},\cdots,t_{p+q})=[L(t_{1},\cdots,t_{p+q})^{n}_+,L(t_{1},\cdots,t_{p+q})]$$
where the subscript $+$ denotes the differential part of the pseudodifferential operator. The operator $P(t_{1},\cdots,t_{p+q})$ is then obtained as $L(t_{1},\cdots,t_{p+q})^{p}$ and the operator $Q(t_{1},\cdots,t_{p+q})$ can be calculated by the equation
$$\frac{\partial }{\partial t_n} Q(t_{1},\cdots,t_{p+q})= [L(t_{1},\cdots,t_{p+q})^{n}_{+},Q(t_{1},\cdots,t_{p+q})]$$

Note that one can construct a larger family of quantum curves solving the same equations where now all the KP times $t_i$ for $i \ge 1$ are allowed to flow. However, in this generality the bi-degree of these curves will not be fixed. If $Q$ is a normalized operator we can define another family of quantum curves $(\hat P(\hat t_{1},\cdots,\hat t_{p+q}),\hat Q(\hat t_{1},\cdots,\hat t_{p+q}))$ solving the KP-equations. These curves are of bi-degree $(q,p)$ and are obtained by deforming the string equation $[Q,-P]=1$ along a second set $\hat t_{1}, \hat t_{2}, \cdots$ of KP times. Our goal is to determine the relation between the two KP orbits by showing in Theorem \ref{main-theorem} that they are related by a type of Fourier duality.

It can be seen that for both KP deformations there are a priori the same amount of relevant KP time variables and this makes it at least conceivable that there is a good correspondence between the two theories. However, the two corresponding KP $\tau$-functions appear to be very different and depend non-trivially on different numbers of the time variables. The beautiful consequence, which is useful also in 2D quantum gravity, is that the duality between the two KP deformations of the quantum curves allows to reduce the number of relevant time variables for one of the two families of quantum curves and hence leads to a simplified description.

To formulate the main result we notice that differential operators with power series coefficients can be regarded as operators acting in the space of polynomials $\mathbb{C}[z]$. Namely, we assume that
$\partial_{x}$ acts as multiplication by $z$ and multiplication by $x$ acts as $-\frac{\textrm{d}}{\textrm{d}z}$. It is clear that this construction gives us a structure of $\bf D$-module on $\mathbb{C}[z]$ where $\textrm{\textbf{D}}=\CC[x,\partial_{x}]$ is the one-variable Weyl algebra. In other words, one obtains a D-module on the plane $\mathbb{A}^{1}$. Using this fact and assuming that the operator $P$ is monic we can assign a connection to a quantum curve $(P, Q)$. To define this connection  we can consider a matrix $M$ of the operator $Q$ in $P$-basis. Here we are saying that the elements $e_1, ..., e_p$ of $\CC[z]$ form a $P$-basis if elements $P^ke_i$ form a basis. If $P$ is monic a $P$-basis exists: we can take, for example, $e_i=z^{i-1}$. In other words, in this case the operator $P$ specifies a structure of free module over the ring of polynomials. The matrix $M$ (the companion matrix of the quantum curve in the terminology of \cite{SCH2}) can be considered as a matrix with polynomial coefficients. It specifies a connection 
$$\nabla = \frac{\textrm{d}}{\textrm{d}u}-M(u)$$
We consider it as a connection on the punctured formal disc $D^{\times}$ centered at the point $\infty$, meaning that we choose a presentation of the punctured disc as $D^{\times} \cong \textrm{Spec } \CC(\!(1/u)\!)$. 

If the operator $Q$ is also monic, we can define another connection using the same construction; we prove that these two connections are related by local Fourier transform. (Note that the local Fourier transform is defined only up to gauge equivalence, hence it would be more accurate to say that the local Fourier transform relates equivalence classes of connections.) Notice that this statement is true also for  quantum curves specified by matrix differential operators.

To state the precise relation obtained by the Fourier transform it is best to work with the whole KP orbits. We prove in Theorem \ref{main-theorem} that the family of connections corresponding to the first family of quantum curves $ (P(t_{1},\cdots,t_{p+q}),Q(t_{1},\cdots,t_{p+q}))$ is related by local Fourier transform with the family of connections corresponding to the family of quantum curves $(\hat P(\hat t_{1},\cdots,\hat t_{p+q}),\hat Q(\hat t_{1},\cdots,\hat t_{p+q}))$ where the parameters $(\hat t_{1},\cdots, \hat t_{p+q})$ are expressed in terms of the parameters $(t_{1},\cdots,t_{p+q})$ in the following way: 

Let $s$ be an indeterminate. For $p,q \ge 1$, we say two functions $f_{1} \in \CC(\!(1/s^{1/p})\!)$ and $f_{2} \in \CC(\!(1/s^{1/q})\!)$ are compositional inverses up to gauge equivalence if 
$$f_{1}^{-1} = f_{2} + \textrm{ element of } s^{-1}\CC[\![s^{-1/q}]\!]$$
$$f_{2}^{-1} = f_{1} + \textrm{ element of } s^{-1}\CC[\![s^{-1/p}]\!]$$
Then the following two functions are inverses up to gauge equivalence:
$$f_{1}=\frac{1}{p} \cdot \sum_{k=1}^{p+q} k t_{k}s^{\frac{k-p}{p}} \;\;\; \textrm{ and } \;\;\; f_{2}=-\frac{1}{q} \cdot \sum_{k=1}^{p+q} k \hat t_{k}s^{\frac{k-q}{q}}$$
This is sufficient to describe each set of parameters in terms of the other one. The gauge ambiguity of the inverse functions comes from the fact that the Fourier transform only relates gauge equivalence classes of connections.

In Section \ref{physics-section} we show that the duality relation between the quantum curves can be applied to give conceptual proofs of the duality in 2D quantum gravity. This duality corresponds to a certain change in the matter content of the theory and is called T -- duality or p -- q duality. The latter notation stems from the fact that for positive co-prime integers $p$ and $q$ the duality relates the partition function of the $(p,q)$ minimal model coupled to gravity to the corresponding function of the $(q,p)$ minimal model coupled to gravity. 

Previous approaches \cite{FKN}, \cite{KM} by physicists were of more computational nature. It seemed clear that a more conceptual approach would be useful. The first step was taken in \cite{LUU}, where we showed how to describe the duality as a local Fourier duality of connections on the formal punctured disc. However, the proofs were still of computational nature. As a consequence of our general duality results for quantum curves obtained in the present work, the reason underlying the precise dynamics of the duality can now be much better understood as consequence of general properties of the local Fourier transform.  

Consider for example the duality between the $(5,2)$ model 
and the $(2,5)$ model of 2D quantum gravity. The former depends a 
priori on the time variables $t_{1}, \cdots ,t_{7}$ and the latter 
depends a priori on the time variables $\hat t_{1},\cdots, \hat t_{7}$. It is known that the $\tau$-function of the $(5,2)$ model 
depends trivially on the variable $t_{5}$ and the $\tau$-function of 
the $(2,5)$ model depends trivially on the variables $\hat t_{2},\hat 
t_{4},\hat t_{6}$. The duality therefore leads to a simplified 
description of the $(5,2)$ model. 

As a consequence of the duality 
the two sets of times can be expressed in terms of each other. For 
example we will give a proof based on the local Fourier transform that the dynamics of the 2 -- 5 duality is given by
\begin{eqnarray*}
t_{1} &=& \hat t_{1}+ \frac{5}{8} \cdot \hat t_{5}^3 - \frac{27}{8} \cdot \hat t_{5}^{2}\cdot \hat t_{6}^{2} + \frac{567}{200} \cdot \hat t_{5}\cdot \hat t_{6}^{4} + \frac{18}{5} \cdot \hat t_{4} \cdot \hat t_{5} \cdot \hat t_{6}  - \frac{3}{2} \cdot \hat t_{3} \cdot \hat t_{5} - \frac{15309}{25000} \cdot \hat t_{6}^6 - \frac{54}{25} \cdot \hat t_{4} \cdot \hat t_{6}^{3} \\
&&\;\;\;\;  + \frac{81}{50}\cdot \hat t_{3} \cdot \hat t_{6}^{2} - \frac{6}{5} \cdot \hat t_{2}\cdot \hat t_{6}  - \frac{4}{5} \cdot \hat t_{4}^{2}\\
\\
t_{2}&=&\hat t_{2}+ 3 \cdot \hat t_{5}^2 \cdot \hat t_{6} -  \frac{108}{25}\cdot \hat t_{5}\cdot\hat t_{6}^{3} - 2 \cdot \hat t_{4} \cdot \hat t_{5} + \frac{3888}{3125} \cdot \hat t_{6}^5 + \frac{72}{25} \cdot \hat t_{4} \cdot \hat t_{6}^{2} - \frac{9}{5} \cdot \hat t_{3} \cdot \hat t_{6}\\
\\
t_{3}&=&\hat t_{3} -\frac{5}{4}\cdot \hat t_{5}^{2}-\frac{12}{5} \cdot\hat t_{4} \cdot \hat t_{6}+\frac{9}{2} \cdot \hat t_{5} \cdot \hat t_{6}^{2}-\frac{189}{100} \cdot \hat t_{6}^{4}\\ 
\\
t_{4}&=& \hat t_{4} -3\cdot \hat t_{5} \cdot \hat t_{6} + \frac{54}{25} \cdot \hat t_{6}^{3}\\
\\
t_{5}&=&\hat t_{5} - \frac{9}{5} \cdot \hat t_{6}^{2}\\
\\
t_{6} &=& \hat t_{6}\\
\\
t_{7}&=& \hat t_{7} \cdot \frac{-5}{2}
\end{eqnarray*}

The shape of such coefficients occurring in relating the time flows in the two dual theories will be deduced from general properties of the local Fourier transform of connections on the formal punctured disc. As is shown in \cite{LUU2}, these are complex analogues of the symmetries underlying the numerical local Langlands duality for $\GL_{n}$ over local fields.

Our results are based on \cite{SCH2}, however, the familiarity with this paper is not necessary for the reading of present paper: all definitions and statements we are using are repeated.

\section{Quantum curves}
Suppose $(P,Q)$ is a bi-degree $(p,q)$ quantum curve. Using the action on $\CC[z]$ described  in the introduction, we can construct two $\D$-modules associated to this curve, meaning two representations  of the Weyl algebra $\textrm{\textbf{D}}=\CC[x,\partial_{x}]$. The first one $\D_{1}(P,Q)$ is given by 
$$\partial_{x} \mapsto P \; \;  \;\textrm { and } \;\; \; x\mapsto Q$$
and the second one $\D_{2}(P,Q)$ is given by
$$\partial_{x} \mapsto -Q \;\;\; \textrm{ and } \; \; \;
x \mapsto  P$$
Hence, the first $\D$-module is the analogue of the second $\D$-module but for the string equation
$[Q, - P]=1$
instead of
$[P, Q]=1$. 
In other words, one has
$$\D_{1}(Q,-P)=\D_{2}(P,Q)$$
The global Fourier transform $\mathcal F^{\textrm{glob}}$ of a $\D$-module is the $\D$-module obtained by composing with the map
$$\partial_{x} \mapsto  x \; \; \textrm{ and } \; \; x \mapsto -\partial_{x}$$
It is clear from the definitions that the two $\D$-modules are globally Fourier dual:
$$\mathcal F^{\textrm{glob}}  \big [ \D_{2}(P,Q) \big ] \cong \D_{1}(P,Q)$$
This global Fourier duality will be crucial in our approach to the duality of KP orbits of quantum curves.

\subsection {KP-flows}
\label{KP-flows-section}
In this section we describe, following \cite{SCH2}, the KP-flows on the space of quantum curves. Via the relevant dressing operators attached to a solution to the string equation $[P,Q]=1$, this is related to the fact that KP-flows can be described geometrically in terms of the Sato Grassmannian. We start by recalling this:

As a set, the big cell $Gr$ of the Sato Grassmannian consists of the $\CC$-subspaces of $\mathcal{H}=\CC(\!(1/z)\!)$ that are comparable to $\mathcal H^{+}=\CC[z]$. This simply means that the projection map to $\mathcal H^{+}$ that takes a Laurent series to its polynomial part is an isomorphism. The algebra of pseudodifferential operators $\CC[\![x]\!](\!(\partial_{x}^{-1})\!)$ acts on the space of Laurent series $\mathcal{H}=\CC(\!(1/z)\!)$ in the following manner: For $f(z)\in \CC(\!(1/z)\!)$ let
$ (x^{m} \partial_{x}^{n})f(z)= \left(-\frac{\textrm{d}}{\textrm{d}z} \right)^{m} (z^{n}f(z))$.  For differential operators this action was described in the introduction. Note that differential operators can be characterized as pseudodifferential operators transforming $\CC[z]$ into itself.

Let us review shortly the elements of Sato theory and the results of \cite {SCH2}. The commutative Lie algebra $\gamma_+$ of polynomials $\sum_{k\geq 0} t_kz^k$ acts on $\mathcal{H}$ by means of multiplication operators, hence it acts on $Gr$ in a natural way. There exists a one-to-one correspondence between the elements of  the group $\mathcal{G}$ of monic zeroth  order pseudodifferential operators and points of $Gr$. Namely, every subspace $V\in Gr$ has a unique representation in the form $V=S\mathcal{H}_+$ where $S\in \mathcal{G}.$
It follows that $\gamma_+$ acts also on $\mathcal{G}$: 
\begin{equation}
\label{s}
\frac{\partial S}{\partial t_n}=(S\partial_{x}^nS^{-1})_-S
\end{equation}
where the subscript ``$-$'' denotes the part of the pseudodifferential operator involving negative powers of $\partial_{x}$.

Suppose now that $P$ is a normalized pseudodifferential operator of order $p$, meaning it is of the form
$$P=\partial_{x}^{p}+a_{p-2}\partial_{x}^{p-2}+\cdots$$
Every such operator can be represented in the form $ S^{-1}\partial_{x}^pS$ where $S\in \mathcal{G}$ and this representation is unique up to multiplication by an operator with constant coefficients. Using this statement  we can construct the action of Lie algebra $\gamma_+$  on the space of normalized pseudodifferential operators differentiating the relation $P(t)= S^{-1}(t)\partial_{x}^{p}S(t)$ with respect  to $t_n$. As indicated in the introduction, the action on this space can be written in the form of differential equations
\begin{equation}
\label{kp}
\frac{\partial P }{\partial t_n}=[P^{\frac{n}{p}}_+,P]
\end{equation}
Notice that this formula determines also the action of Lie algebra  $\gamma_+$ on the space of normalized differential operators. All actions we described  can be considered as different forms of the KP hierarchy. Integrating the actions of Lie algebra $\gamma_+$
we obtain the action of a commutative group $\Gamma_+$, the elements of this group can be written in the form $g(t)=\exp(\sum_{k\geq 0} t_kz^k).$

We now come back to our aim of describing the KP-flows on the space of quantum curves.
We would like to solve the string equation $[P,Q]=1$. We will recall the relevant results from \cite{SCH2}, but in contrast to loc. cit. we will assume that $P$ is a normalized differential operator instead of $Q$. One can say that we apply the results of \cite{SCH2} to the string equation $[-Q,P]=1$. This normalization is useful when working with the $\D$-module $\D_{2}(P,Q)$ that we defined earlier. 

As we mentioned, we can construct the operator $S\in \mathcal{G}$ such that $SPS^{-1}=\partial_{x}^p$. Introducing the notation
$V=S\mathcal{H}_+$ we obtain a subspace $V\in Gr$ invariant with respect to multiplication by $z^p$ and with respect to the action of the operator $-\tilde Q= \frac{1}{pz^{p-1}}\frac{\textrm{d}}{\textrm{d}z}+b(z)$ where $b(z)$ stands for the multiplication by a Laurent series denoted by the same letter. Here we use the fact that the action of $\partial_{x}^p$ can be interpreted as multiplication by $z^p$ and the fact that $\mathcal{H}_+$ is invariant with respect to the action of differential operators. The form of the operator $-\tilde Q=S(-Q)S^{-1}$ follows from the relation $[-\tilde Q,\tilde P]=1$ where $\tilde P=SPS^{-1}=\partial_{x}^p$. We can invert this consideration to obtain the following statement, see \cite {SCH2}: If the point $V$ of the big cell of the Grassmannian satisfies 
$$z^p V \subseteq V  \;\;\; \textrm{ and } \;\;\; \left(\frac{1}{pz^{p-1}}\frac{\textrm{d}}{\textrm{d}z}+b(z)\right) V \subseteq V$$ 
we can construct a differential operator $P$ and a normalized differential operator $Q$ obeying $[P,Q]=1$. The leading term of the operator $Q$ is determined by the leading term of the Laurent series $b(z)$. 

To summarize, finding a solution to the string equation is equivalent to finding a suitable point of the Sato Grassmannian stabilized by two operators of the above form.

It follows that the Lie algebra $\gamma _+$ acts  on the space of pairs $(P,Q)$ of differential operators obeying $[P,Q]=1$ under the assumption that $P$ is a normalized operator. The proof is based on the remark that for $V$ as above, the space 
$$V(t):=g(t)V  \;\;\; \textrm{ where } \;\;\; g(t)=\exp \left (\sum_{k\geq 0} t_kz^k \right )\in \Gamma_+$$
satisfies a similar condition as $V$ but with $-\tilde Q$ replaced by $g^{-1}(t)(-\tilde Q) g(t)$: Namely, one has
\begin{equation}
\label{b}
z^p V(t) \subseteq V(t)  \;\;\; \textrm{ and } \;\;\; \left(\frac{1}{pz^{p-1}}\frac{\textrm{d}}{\textrm{d}z}+b(z)-\sum \frac{k}{p}t_kz^{k-p} \right) V(t) \subseteq V(t)
\end{equation}
It follows that the KP-flows (\ref {kp}) are defined on the space of quantum curves. 

To simplify the formulas that we will obtain in Theorem \ref{main-theorem}, we normalize the KP times so that at times $t_{1},t_{2},\cdots$ the point $V(t)$ of the KP orbit of $V$ is stabilized by $z^{p}$ and $\frac{1}{pz^{p-1}}\frac{\textrm{d}}{\textrm{d}z}-\sum \frac{k}{p}t_kz^{k-p}$.  

\subsection{Companion matrix connections}
We now describe in detail how to attach a connection on the formal punctured disc to a quantum curve. Recall first that the category $\textrm{Conn}(D^{\times})$ of connections on the formal punctured disc $D^{\times}$ can be defined as follows: After choosing an isomorphism $D^{\times} \cong \Spec \; \CC(\!(t)\!)$ each object in this category can be described as a pair $(M,\nabla)$ where $M$ is a finite dimensional $\CC(\!(t)\!)$-vector space and $\nabla$ is  a $\CC$-linear map 
$$\nabla : M \longrightarrow M \;\;\; \textrm{ such that } \;\;\; \nabla(f \cdot m) = f \cdot \nabla(m) + \frac{\textrm{d} f}{\textrm{d} t} \cdot m$$
for all $f \in \CC(\!(t)\!)$ and all $m \in M$. The morphisms in the category $\textrm{Conn}(D^{\times})$ are $\CC$-linear maps that also commute with the maps $\nabla$. 

The category $\textrm{D-Mod}$ of left ${\bf D}=\CC[x,\partial_{x}]$-modules can be viewed as the category of $\D$-modules over $\PP^{1}(\CC) \backslash \{\infty\}$. Hence the $\D$-modules $\D_{1}$ and $\D_{2}$ that we associated to a quantum curve are $\D$-modules on the plane $\PP^{1}(\CC)\backslash \{\infty\}$.

Let $\textrm{Hol}(\textrm{D-Mod})$ denote the full sub-category of holonomic $\D$-modules on  $\PP^{1}(\CC)\backslash \{\infty\}$. There is a restriction functor that captures the local information near $\infty$:
$$\psi_{\infty} : \textrm{Hol}(\textrm{D-Mod}) \longrightarrow \textrm{Conn}(D^{\times})$$
It is given on objects by
$$N \longmapsto N \otimes_{\CC[x]} \CC(\!(t)\!) \; \; \;\textrm{ with } \;\;\; t=\frac{1}{x}$$
where the $\CC(\!(t)\!)$-vector space structure comes from the second factor and the map
$$\nabla : \psi_{\infty} N \longrightarrow \psi_{\infty} N$$
is given by
$$n \otimes f \mapsto (\partial_{x} \cdot n) \otimes (-\frac{1}{t^{2}} f) + n \otimes \frac{\textrm{d}}{\textrm{d}t} f$$ 

Applying the functor $\psi_{\infty}$ to $\D$-modules related by the global Fourier transform we obtain connections related by local Fourier transform. This statement could be considered as a definition of the local Fourier transform for connections that can be obtained from $\D$-modules by means of the functor  $\psi_{\infty}$. See the discussion preceding Theorem \ref{main-theorem} for a precise statement of the result.

Let us consider  a solution to the string equation $[P,Q]=1$ with $P$ normalized. For $SPS^{-1}=\partial_{x}^p$ let $V=S\mathcal{H}_+$ be the associated point of the Grassmannian. Suppose $e_0,...,e_{p-1}$ is a $P$-basis of  $\mathcal{H}_+$, meaning that the collection of elements $P^{k}e_{i}$ for varying $k$ and $i$ is a $\CC$-basis. Under this assumption one has that $v_0, ....,v_{p-1}$ with $v_i=Se_i$ is a $z^p$-basis of $V$.  
We obtain
\begin{equation}
\label{ }
\tilde Pv_i=M^j_i(z^p)v_j
\end{equation}
where the entries of the matrix $M$ are polynomials with respect to $z^p$.  The matrix $M$  coincides with  the companion matrix of the pair $(P,Q)$ defined as a matrix of $P$ in $Q$-basis. 

We apply these constructions to the $\D$-modules coming from a bi-degree $(p,q)$ quantum curve $(P,Q)$ where $P$ is monic. The D-module $\D_{2}(P,Q)$ is holonomic and we can apply the functor $\psi_{\infty}$ to obtain a $p$-dimensional connection. One sees that
$$\mathcal B := \{1\otimes 1, \cdots, z^{p-1} \otimes 1\}$$
is a $\CC(\!(t)\!)$-basis of the vector space $\D_{2}(P,Q) \otimes_{\CC[x]} \CC(\!(t)\!)$. The connection is nothing but the companion matrix connection introduced in \cite{SCH2} and we will denote it by $\nabla_{M(P,Q)}$. 

Since on $\D_{2}(P,Q)$ we defined the $\partial_{x}$ action to be the action of $-Q$, it follows that apart from the factor $-1/t^{2}$, the matrix of the action of $\partial_{x}$ with respect to the basis $\mathcal B$  is the matrix of the action of $-Q$ with respect to the $P$-basis $\{1, z, \cdots, z^{p-1}\}$ of $\CC[z]$. Hence this is simply $-M(P,Q)$. Hence, by writing the connection with respect to the basis $\mathcal B$, one can write
$$\D_{2}(P,Q) \otimes_{\CC[x]} \CC(\!(t)\!) \cong \left (\CC(\!(t)\!)^{p}, \frac{\textrm{d}}{\textrm{d}t} -M(P,Q)(1/t)\cdot (-1/t^{2}) \right )$$
Or, more naturally, when we write this connection in terms of the coordinate $u$ such that the punctured disc is given by $D^{\times} \cong \textrm{Spec }\CC(\!(1/u)\!)$ the connection is simply described by
$$\frac{\textrm{d}}{\textrm{d}u}-M(P,Q)$$
Hence, this is the companion matrix connection $\nabla_{M(P,Q)}$.  Furthermore,  if $p$ and $q$ are co-prime it is known, see \cite{SCH2}, that one obtains an irreducible connection.

\subsection{Classification of connections}
The isomorphism classes of connections have been classified in work of Levelt and Turrittin, see for example \cite{BV} for a detailed exposition. Therefore, one can ask what the description of the connection of the quantum curve in terms of this classification is. In order to answer this, we describe the Levelt-Turrittin classification. 

Sometimes it is convenient to describe a connection with respect to a choice of basis of the vector space $M$. Then the map $\nabla$ is simply of the form
$$\frac{\textrm{d}}{\textrm{d}t}+ A(t)$$
with $A(t) \in \mathfrak g \mathfrak l_{n}\CC(\!(t)\!)$. This classification implies in particular that given an irreducible connection one can simplify the connection matrix $A(t)$ if one changes coefficients from $\CC(\!(t)\!)$ to some suitable finite extension $\CC(\!(t^{1/q})\!)$. For the extended coefficients there is $g \in \GL_{n}(\CC(\!(t^{1/q})\!))$ such that with respect to the new basis obtained by multiplying by $g$ the old one, the description of the connection becomes
$$\frac{\textrm{d}}{\textrm{d}t}+ gA(t)g^{-1} +g \cdot  \frac{\textrm{d}}{\textrm{d}t} g^{-1}$$
with
$$gA(t)g^{-1} + g \cdot \frac{\textrm{d}}{\textrm{d}t} g^{-1} =  \begin{bmatrix}
f_{1} & & &\\
&\ddots &&\\
&&f_{n}
\end{bmatrix}$$
for some functions $f_{i} \in \CC(\!(t^{1/q})\!)$. This type of simplification is conveniently described in terms of the push-forward and pull-back of connections along suitable maps. Given a map 
$$\rho : \CC[\![t]\!] \longrightarrow \CC[\![u]\!]$$
which takes $t$ to some element in $u \CC[\![u]\!]$ one can define associated push-forward and pull-back operations on the categories of connections on the formal punctured disc with local coordinate $t$ and $u$ respectively. Following the notation in \cite{FAN} we denote by $[i]$ for $i \in \ZZ^{\ge 1}$ the map $\rho$ that takes $t$ to $u^{i}$. If $N$ is a $d$-dimensional connection over $\CC(\!(u)\!)$ then $[i]_{*}N$ is a $d \cdot i$-dimensional connection over $\CC(\!(t)\!)$. If $M$ is a $d$-dimensional connection over $\CC(\!(t)\!)$ then $[i]^{*}M$ is a $d$-dimensional connection over $\CC(\!(u)\!)$. 

Now suppose $f\in \CC(\!(t^{1/q})\!)$ and $q$ is the minimal positive such integer. The connection denoted by $E_{f}$ in \cite{GRA} denotes the push-forward along the map $\CC(\!(t^{1/q})\!) \longrightarrow \CC(\!(t)\!)$ that takes $t^{1/q}$ to $t$ of the one-dimensional connection
$$\frac{\textrm{d}}{\textrm{d}t^{1/q}}+ qt^{(q-1)/q}f(t^{1/q})$$
over $\CC(\!(t^{1/q})\!)$. The classification result is then that every irreducible connection is isomorphic to some $E_{f}$ and $f$ is uniquely determined up to adding an arbitrary element in
$$\frac{1}{q}\cdot \ZZ + t^{1/q}\CC[\![t^{1/q}]\!]$$ 
We will also use Fang's results \cite{FAN} later on and there a different notation is used, so we introduce it now: 

Given a function $\alpha \in \CC(\!(t)\!)$ one denotes by $[\alpha]$ the connection 
$$(\CC(\!(t)\!), \frac{\textrm{d}}{\textrm{d}t}+ \frac{\alpha}{t})$$
The Levelt-Turrittin classification in this language implies that every irreducible connection is isomorphic to one of the form 
$$[q]_{*}([t\partial_{t}(\alpha)] \otimes_{\CC(\!(t)\!)} R)$$
where $R$ is a suitable regular connection
$$R=(\CC(\!(t)\!), \frac{\textrm{d}}{\textrm{d}t}+\frac{r}{t}) \;\;\; \textrm{ for some }r \in \CC$$
and the tensor product of two connections $(V_{i},\nabla_{i})$ has underlying vector space $V_{1} \otimes_{\CC(\!(t)\!)} V_{2}$ and the connection is given via
$$v_{1} \otimes v_{2} \mapsto \nabla_{1}(v_{1}) \otimes v_{2} + v_{1} \otimes \nabla_{2}(v_{2})$$
for all $v_{1} \in V_{1}$ and $v_{2} \in V_{2}$.

In order to conform with usual notational conventions for Kac-Schwarz operators, it will often be useful to express connections on the formal punctured disc with coordinate $u$ in terms of the reciprocal coordinate $z=1/u$.  Hence, for $h(z) \in \CC(\!(1/z)\!)$ we use the notation
$$(\CC(\!(u)\!), \frac{
\textrm{d}}{\textrm{d} z} + h(z)) := (\CC(\!(u)\!), \frac{
\textrm{d}}{\textrm{d} u} - \frac{1}{u^{2}} h(1/u))$$
Note that in the following we will write $u$ for $1/z$ without further comment.

\subsection{Companion matrices and Levelt-Turrittin normal form}
\label{KP-normal-section}
Our goal is to analyze the behavior of KP-flows under duality $(P,Q)\to (-Q,P).$ We know already that on the connection corresponding to the companion matrix $M$ this duality acts as local Fourier transform. From the other side we know  the action of KP-flows on $b$ (see (\ref{b})).
 
We now describe the relation between $b(z)$ and $M$ in the case  when $P$ is normalized assuming that $p$ and $q$ are co-prime. In this case the companion matrix connection $\nabla_{M(P,Q)}$ is irreducible. It follows from the Levelt-Turrittin classification and the irreducibility of $\nabla_{M(P,Q)}$ that
$$[p]^{*}\nabla_{M(P,Q)} \cong (\CC(\!(u)\!)^{p}, \frac{\textrm{d}}{\textrm{d}z}  + \begin{bmatrix}
\xi_{1}(z) &&& \\
&\ddots && \\
&&&\xi_{p}(z)
\end{bmatrix} )$$
for suitable $\xi_{i} \in \CC(\!(z)\!)$ that satisfy
$$\xi_{i}(z) \in \frac{1}{z}\CC[\![z]\!]$$
It is shown in \cite{SCH} that up to a $p$'th root of unity the coefficient of $1/z$ is given by $(1-p)/2$.
Moreover, it is known that the gauge transformation can be taken to be of the form
$$R(z) \in \GL_{p}(\CC[\![1/z ]\!])$$
Let $R_{i,j}(z)$ denote the $(i+1,j+1)$ entry of $R(z)$ and let
$$t_{i}(z)=  \exp \left (\int b(z) pz^{p-1} \textrm{d}z \right ) \sum_{j=0}^{p-1} R_{i,j}(z)  \frac{v_{j}}{z^{j}} $$
One obtains that
$$\frac{1}{pz^{p-1}}\frac{\textrm{d}}{\textrm{d}z}  t_{i}(z) = \Lambda_{i}(z) t_{i}(z) \;\;\; \textrm{ where } \;\;\; \Lambda_{i}(z) = \frac{\xi_{i}(z)}{pz^{p-1}}$$
Let the constants $c_{i}$ be such that
$$t_{i}(z)=c_{i} \exp(\int \Lambda_{i}(z)pz^{p-1} \textrm{ d}z)$$
Since
$$v_{i}=z^{i} + \textrm{ lower order terms }$$
it follows that
$$\begin{bmatrix}
c_{1} \exp(\int (\Lambda_{1}(z) -b(z))pz^{p-1}\textrm{ d}z) \\
\vdots\\
c_{n} \exp(\int (\Lambda_{n}(z) -b(z))pz^{p-1}\textrm{ d}z)
\end{bmatrix} = R\cdot \begin{bmatrix}
1 + \mu_{1,1}z^{-1}+ \cdots \\
\vdots\\
1+\mu_{n,1}z^{-1}+\cdots
\end{bmatrix}
$$
for suitable constants $\mu_{i,j}$. Note that since $R$ is invertible not all $c_{i}$'s can vanish. It also follows from the above equation that for each $i$ such that $c_{i} \ne 0$ one needs
$$\Lambda_{i}(z)-b(z) \in \frac{1}{z^{p+1}} \CC[\![1/z]\!]$$
Since 
$\Lambda_{i}(z) \in \frac{1}{z^{p}} \CC[\![z]\!]$
and since
$b(z) \in \frac{1}{z^{p}} \CC[\![z]\!]$
it follows that
$$b(z)=\Lambda_{i}(z)$$
This holds for all $i$ such that $c_{i}\ne 0$. Now define a one-dimensional connection over $\CC(\!(u)\!)$ by
$$\nabla_{\textrm{KS}}:= (\CC(\!(u)\!), \frac{\textrm{d}}{\textrm{d}z} + \tilde b(z)) \;\;\; \textrm{ where } \;\;\;
\frac{1}{pz^{p-1}}\tilde b(z) = b(z)$$
It follows from the previous calculations that
$$\Hom_{\CC(\!(u)\!)}(\nabla_{\textrm{KS}},[p]^{*}\nabla_{M(P,Q)}) \ne (0)$$ 
Furthermore, one can define a connection on this space of homomorphisms and via the projection formula for pull-back and push-forward, see for example \cite{SAB} (Section 1), one has
$$[p]_{*} \Hom_{\CC(\!(u)\!)}(\nabla_{\textrm{KS}},[p]^{*}\nabla_{M(P,Q)}) \cong \Hom_{\CC(\!(u^{p})\!)}([p]_{*}\nabla_{\textrm{KS}}, \nabla_{M(P,Q)})$$ 
Since $\nabla_{M(P,Q)}$ is irreducible it follows that
$$\nabla_{M(P,Q)}\cong [p]_{*} \nabla_{\textrm{KS}}$$
as desired.

It can be seen from the description of the KP flows in Section \ref{KP-flows-section} and the calculations in \cite{SCH} that the KP times $t_{1},\cdots,t_{p+q}$ can be normalized so that 
$$ \nabla_{M(P(t_{1},\cdots,t_{p+q}),Q(t_{1},\cdots,t_{p+q}))}\cong [p]_{*}\left( (\CC(\!(u)\!), \frac{\textrm{d}}{\textrm{d}z} - pz^{p-1}(\frac{1-p}{2p} \frac{1}{z^{p}}+ \frac{1}{p}\sum_{i=1}^{p+q} i t_{i} z^{i-p})) \right )$$
In this normalization, the connection only depends on $p$ and $q$ and we will denote it by $\nabla_{M_{p,q}(t_{1},\cdots,t_{p+q})}$.

\section{Duality of quantum curves}
As explained previously, if $P$ is normalized one can deform a bi-degree $(p,q)$ quantum curve $[P,Q]=1$ along the first $p+q$ flows of the KP hierarchy.  If $Q$ is normalized we have another deformation coming from the quantum curve $[-Q,P]=1. $ We would like to consider the case when after the deformation both $P$ and $Q$ remain normalized. As can be seen from the arguments in Section \ref{KP-flows-section}, this means in particular that no time evolution is taking place along the flows of the $p+q$'th and $p+q-1$'th KP times. More precisely, as can be seen from the formula at the end of Section \ref{KP-flows-section}, this formally corresponds to the following constraints for the KP times $t_{1},\cdots,t_{p+q}$ of the first string equation and the KP times $\hat t_{1},\cdots, \hat t_{p+q}$ of the second string equation:
$$t_{p+q}= \frac{p}{p+q}=-\frac{p}{q} \cdot \hat t_{p+q}$$
$$t_{p+q-1}=0=\hat t_{p+q-1}$$
Our aim is to relate the companion matrix connections $\nabla_{M_{p,q}(t_{1},\cdots,t_{p+q})}$ and $\nabla_{M_{q,p}(\hat t_{1},\cdots, \hat t_{p+q})}$ associated with the two deformations.

Our main tool is the local Fourier transform: Bloch and Esnault \cite{BE} and Lopez \cite{LOP} developed an analogue over $\CC$ of the $\ell$-adic local Fourier transform constructed by Laumon in \cite{LAU}. See also the work of Arinkin \cite{ARI} for an alternate approach. 

The stationary phase principle for the local Fourier transform is crucial for our proof of the main theorem. This is a central part of the theory of local Fourier functors. In fact, from its very introduction in the $\ell$-adic context, the guidance in defining the local Fourier transform is the search for a stationary phase principle for the global Fourier transform. We now describe the relevant results. 

Fix a point $\infty \in \PP^{1}(\CC)$ and let $(M,\nabla)$ be a  connection on the formal punctured disc centered at $\infty$. A Katz extension $\mathcal M$ of this connection is a certain connection on the punctured plane $\PP^{1}(\CC)\backslash \{0,\infty\}$ which is regular singular at $0$. See for example \cite{BBDE}, Theorem 2.8, for a precise definition of this extension. For the connections of interest to our considerations, namely the irreducible objects of $\textrm{Conn}(D^{\times})$ this can be described in the following manner: Consider the connection 
$$E_{f} \;\;\; \textrm{ with } \;\;\; f= \sum_{ i \gg -\infty} a_{i} t^{i/p}$$
The Katz extension $\mathcal E_{f}$ of this connection is then simply given by 
$$\frac{\textrm{d}}{\textrm{d} t} + \sum_{ -1 \ge i \gg -\infty} a_{i} t^{i/p}$$
It is clear from the description of Katz extension in \cite{BBDE} 
that a Katz extension of the previously introduced companion matrix connection $\nabla_{M(P,Q)}$ is given by $\D_{2}(P,Q)$.

Let $\textrm{Conn}(D^{\times})^{>1}$ denote the full subcategory of $\textrm{Conn}(D^{\times})$ consisting of connections with slopes strictly bigger than $1$. The condition $\textrm{slope}(E_{f}) >1$ simply means that there is an $i < -p$ with $a_{i}\ne 0$. 
The local Fourier transform is then a functor
$$\mathcal F^{(\infty,\infty)}: \textrm{Conn}(D^{\times})^{>1} \longrightarrow \textrm{Conn}(D^{\times})^{>1}$$
defined in the following way: Consider the Katz extension  of the connection as a $\D$-module, take global Fourier, apply the functor $\psi_{\infty}$. 
In particular, see (\cite{BE}, Proposition 3.12 (v)),
$$\psi_{\infty} \left( \mathcal F^{\textrm{glob}}(\mathcal E_{f}) \right )\cong \mathcal F^{(\infty,\infty)}(E_{f})$$
This is an incarnation of the stationary phase principle for the Fourier transform and it is a key tool in proving the following theorem. Note that this result relates two connections up to gauge equivalence. Furthermore, it is useful to recall the following definition from the introduction: For $p,q \ge 1$, $f_{1} \in \CC(\!(1/s^{1/p})\!)$ and $f_{2} \in \CC(\!(1/s^{1/q})\!)$ are compositional inverses up to gauge equivalence if 
$f_{1}^{-1} = f_{2} + \textrm{ element of } s^{-1}\CC[\![s^{-1/q}]\!]$ and
$f_{2}^{-1} = f_{1} + \textrm{ element of } s^{-1}\CC[\![s^{-1/p}]\!]$. For the statement of the following result recall that we defined the suitably normalized family of connections $\nabla_{M_{p,q}(t_{1},\cdots,t_{p+q})}$ associated to differential operators $P$ and $Q$ at the end of Section \ref{KP-normal-section}.
\begin{thm}
\label{main-theorem}
Let $p$ and $q$ be positive co-prime integers and consider a quantum curve $(P,Q)$ of bi-degree $(p,q)$ with $P$ and $Q$ normalized. Let $t_{1},t_{2},\cdots$ and $\hat t_{1},\hat t_{2},\cdots$ be two sets of KP times such that
$$t_{p+q}= \frac{p}{p+q}=-\frac{p}{q} \cdot \hat t_{p+q}$$
and
$$t_{p+q-1}=0=\hat t_{p+q-1}$$
Then 
\begin{equation}
\label{mt}
\mathcal F^{(\infty,\infty)} \nabla_{M_{p,q}(t_{1},\cdots,t_{p+q})} \cong \nabla_{M_{q,p}(\hat t_{1}, \hat t_{2},\cdots, \hat t_{p+q})}
\end{equation}
where the relation between the times $t_{1},\cdots, t_{p+q}$ and $\hat t_{1}, \cdots, \hat t_{p+q}$ is given by the fact that the following two functions are inverses up to gauge equivalence:
$$f_{1}=\frac{1}{p} \cdot \sum_{k=1}^{p+q} k t_{k}s^{\frac{k-p}{p}}$$
and
$$f_{2}=-\frac{1}{q} \cdot \sum_{k=1}^{p+q} k \hat t_{k}s^{\frac{k-q}{q}}$$
where $s$ is an indeterminate and $f_{1}$ and $f_{2}$ are viewed as elements of $\CC(\!(1/s^{1/p})\!)$ and $\CC(\!(1/s^{1/q})\!)$.
\end{thm}
\begin{proof}
The relation between $t_{p+q}$ and $\hat t_{p+q}$ has already been shown. We now prove the remaining parts of the theorem. From the previously mentioned global Fourier duality
$$\mathcal F^{\textrm{glob}} \big [  \D_{2}(P(t_{1}, t_{2},\cdots, t_{p+q}),Q(t_{1}, t_{2},\cdots, t_{p+q})) \big ] \cong \D_{1}(P(t_{1}, t_{2},\cdots, t_{p+q}), Q(t_{1}, t_{2},\cdots, t_{p+q}))$$
it follows from the stationary phase principle that
$$
\psi_{\infty} \D_{1}(P(t_{1}, t_{2},\cdots, t_{p+q}), Q(t_{1}, t_{2},\cdots, t_{p+q})) \cong  \mathcal F^{(\infty,\infty)}  \nabla_{M_{p,q}(t_{1},\cdots,t_{p+q})}  $$ 
Furthermore, by the discussion of Section \ref{KP-normal-section}, one also has
$$
\psi_{\infty} \D_{1}(P(t_{1}, t_{2},\cdots, t_{p+q}), Q(t_{1}, t_{2},\cdots, t_{p+q}))  \cong  \nabla_{M_{q,p}(\hat t_{1},\cdots, \hat t_{p+q})}$$
for some choices of $\hat t_{1},\cdots, \hat t_{p+q}$. We now relate these KP times concretely to the other set of KP times $t_{1},\cdots, t_{p+q}$. To do so, we formulate the explicit description of the local Fourier transform as obtained by Fang \cite{FAN}, Graham-Squire \cite{GRA}, and Sabbah \cite{SAB}. 

Let us first describe Fang's version of the result. For $Z=t^{-1/p}$ and $Z'=t'^{-1/q}$ and $\alpha \in \CC(\!(Z)\!)$ it follows from \cite{FAN} (Theorem 1.3) that for every regular one-dimensional connection $R$ one obtains
$$\mathcal F^{(\infty,\infty)} \left( [p]_{*}  \left( [Z\partial_{Z}(\alpha)] \otimes_{\CC(\!(Z)\!)} R \right) \right )\cong [q]_{*} \left( [ Z'\partial_{Z'}(\beta) + \frac{p+q}{2} ]\otimes_{\CC(\!(Z')\!)} R  \right )$$
where $\alpha$ and $\beta$ are related in the following manner:
$$\partial_{t} \alpha + t'=0$$
$$\alpha + tt'=\beta$$
Note that since $p+q>p$ the connection to which we apply the local Fourier transform has slope strictly bigger than one and so the results of \cite{FAN} do indeed apply. It follows from the above equations that
$$\partial_{t'} \beta = \partial_{t} \alpha \cdot \partial_{t'}t+ t+  t'\partial_{t'}t = (\partial_{t} \alpha + t') \partial_{t'}t + t=t$$
and
$$\partial_{t'} \beta = (-\partial_{t} \alpha)^{-1}(t')$$
where we consider the compositional inverse of $-\partial_{t}\alpha$ viewed as an element in $\CC(\!(1/t^{1/p})\!)$. Note that for example by \cite{GRA} (Lemma 5.1) such an inverse does indeed exists as a formal Laurent series. We now use this to obtain the time dynamics. 

For this purpose, only the irregular part of the relevant connections matters. This part is defined in the following manner:
For an irreducible connection $E_{f}$ with $f =\sum a_{i}t^{i/p}$ the sum $\sum_{i < 0} a_{i}t^{i/p}$ only depends on the isomorphism class of $E_{f}$. Therefore one can define in a well-defined manner the associated connection
$$(E_{f})_{\textrm{irreg}} := E_{\tilde f} \; \; \textrm{ where } \; \; \tilde f = \sum_{i <0} a_{i}t^{i/p}$$
Note that for $z=1/Z$ one has
$$[Z\partial_{Z}(\alpha)]=(\CC(\!(u)\!),\frac{\textrm{d}}{\textrm{d}z} -\frac{\partial_{Z}(\alpha)}{z^{2}} )$$
Hence, if a function $H$ satisfies
$$pz^{p-1}H(z)= -\frac{\partial_{Z}(\alpha)}{z^{2}}$$
then
$$\partial_{t}\alpha=H$$
It follows that the explicit form of the local Fourier transform implies
$$\left ( \mathcal F^{(\infty,\infty)} \left ( [p]_{*}(\CC(\!(u)\!),\frac{\textrm{d}}{\textrm{d} z} + pz^{p-1} H )\right ) \right )_{\textrm{irreg}} \cong \left( [q]_{*}(\CC(\!(u)\!),\frac{\textrm{d}}{\textrm{d} z} + qz^{q-1} (-H)^{-1})\right)_{\textrm{irreg}}$$
where $(-H)^{-1}$ denotes the compositional inverse of $-H$ viewed as a function of $1/z$, hence as a function of $1/s^{1/p}$ where
$$s:=z^{p}$$
The above described explicit formula for the local Fourier transform can also be verified by comparison to the formulas given by Graham-Squire in \cite{GRA} which were obtained by a different method than Fang's. Let $f$ be such that
$$E_{f} \cong  [p]_{*}  \left( [Z\partial_{Z}(\alpha)] \otimes_{\CC(\!(Z)\!)} R \right)$$
One has, compare to the discussion in \cite{GRA} (Section 5.1), that $z=t$, $\hat z=t'$ and
$$-t\partial_{t}\alpha=\frac{1}{t}\partial_{1/t} \alpha = f$$
and
$$-t'\partial_{t'} \beta=\frac{1}{t'}\partial_{1/t'} \beta = g  + \textrm{ regular terms}$$
Hence one obtains
$$f=z \hat z$$
which is the first of the equations obtained by Graham-Squire. Furthermore, one has
$$z=(\frac{f}{z})^{-1}(\hat z)=- \frac{g}{\hat z} + \textrm{ regular terms}$$
and this yields
$$g=-f + \textrm{ regular terms}$$
and this is the second of the two equations given in \cite{GRA}.

We now apply these generalities concerning the local Fourier transform to our concrete situation: Since $p$ and $q$ are co-prime, at least one of the numbers $p$ and $q$ is odd. Assume first that $p$ is odd. Recall that
$$\nabla_{M_{p,q}(t_{1},\cdots,t_{p+q})} \cong [p]_{*}\left( (\CC(\!(u)\!), \frac{\textrm{d}}{\textrm{d}z} - pz^{p-1}(
\frac{1-p}{2p}\frac{1}{z^{p}}+ \frac{1}{p}\sum_{i=1}^{p+q} i t_{i} z^{i-p})) \right )$$
Hence, since
$$\frac{1-p}{2p} \in \frac{1}{p}\ZZ,$$
it follows that
\begin{eqnarray*}
\psi_{\infty} \D_{1}(P(t_{1}, t_{2},\cdots, t_{p+q}), Q(t_{1}, t_{2},\cdots, t_{p+q}))  &\cong &  \nabla_{M_{q,p}(\hat t_{1},\cdots, \hat t_{p+q})} \\
& \cong & \mathcal F^{(\infty,\infty)}  [p]_{*}\left( (\CC(\!(u)\!), \frac{\textrm{d}}{\textrm{d}z} + pz^{p-1}(- \frac{1}{p}\sum_{i=1}^{p+q} i t_{i} z^{i-p})) \right )
\end{eqnarray*}
Therefore, if we define
$$H(s):= -\frac{1}{p}\sum_{i=1}^{p+q} i t_{i} s^{(i-p)/p},$$
then the coefficent of $1/s$ in $-H(s)$ is zero. It then follows for example from \cite{GRA} (Lemma 5.3) that the coefficient of $1/s$ in $(-H(s))^{-1}$ is zero as well. Therefore there are isomorphisms
$$\xymatrixrowsep{3.2pc}\xymatrix{  [q]_{*}\left( (\CC(\!(u)\!), \frac{\textrm{d}}{\textrm{d}z}  + qz^{q-1}(\frac{1}{p}\sum_{i=1}^{p+q} i
t_{i} s^{(i-p)/p})^{-1} \right )  \\
\left (  \mathcal F^{(\infty,\infty)}  [p]_{*}\left( (\CC(\!(u)\!), 
\frac{\textrm{d}}{\textrm{d}z} + pz^{p-1}(- \frac{1}{p}\sum_{i=1}
^{p+q} i t_{i} s^{(i-p)/p})) \right )   \right )_{\textrm{irreg}}  
\ar[u]_-{\cong} \ar[d]^-{\cong}    \\
[q]_{*}\left( (\CC(\!(u)\!), 
\frac{\textrm{d}}{\textrm{d}z} + qz^{q-1}(- \frac{1}{q}\sum_{i=1}
^{p+q} i \hat t_{i} s^{(i-q)/q})) \right ) }$$
Hence the following two functions are compositional inverses up to gauge equivalence:
$$\frac{1}{p}\sum_{k=1}^{p+q} k  t_{k} s^{\frac{k-p}{p}} \;\;\; \textrm{ and } \; \; \;  - \frac{1}{q}\sum_{k=1}^{p+q} k \hat t_{k} s^{\frac{k-q}{q}}$$
Assume now that $q$ is odd. 

We first recall a general result about the local Fourier transform: Let $a,b$ be indeterminates and denote by $\iota$ the pull-back map of connections on the formal punctured disc along the map
$$\CC(\!(a)\!) \longrightarrow \CC(\!(b)\!) \; \; \; \textrm{ with } \;\;\; a \mapsto -b $$
Suppose given a connection $(M,\nabla)$. The pull back is then given by the $\CC(\!(a)\!)$-vector space $\CC(\!(a)\!) \otimes_{\CC(\!(b)\!)} M$ with the connection map that satisfies
$$1\otimes m \mapsto - 1 \otimes \nabla(m)$$ 
Therefore one obtains
$$\iota \left ([p]_{*}(\CC(\!(u)\!),\frac{\textrm{d}}{\textrm{d}z}+h(z)) \right) \cong [p]_{*}(\CC(\!(u)\!),\frac{\textrm{d}}{\textrm{d}z}-h(\zeta z))$$
where 
$$\zeta^{p}=-1$$
The involution $\iota$ is related to the local Fourier transform: By \cite{BE} (Proposition 3.12 (iv)) one has 
$$\mathcal F^{(\infty,\infty)} \circ \mathcal F^{(\infty,\infty)} = \iota$$ 
It follows that there are isomorphisms
$$\xymatrixrowsep{3.2pc}\xymatrix{ \mathcal F^{(\infty,\infty)}  \nabla_{M_{q,p}(\hat{t}_{1},\hat{t}_{2},\cdots,\hat t_{p+q})} \ar[d]^-{\cong}\\
\iota \nabla_{M_{p,q}(t_{1},t_{2},\cdots, t_{p+q})} \ar[d]^-{\cong} \\
[p]_{*} \left( \CC(\!(u)\!), \frac{\textrm{d}}{\textrm{d}z}  + pz^{p-1}(\frac{1}{p}\sum_{i=1}^{p+q} i t_{i} (\zeta z)^{i-p}))\right ) }$$
where $\zeta^{p}=-1$.
Here for the second isomorphism we have used that $q$ is odd.
It now follows by a similar reasoning as before that the following two functions are compositional inverses up to gauge equivalence:
$$\frac{1}{p}\sum_{k=1}^{p+q} k  t_{k} (-s)^{\frac{k-p}{p}} \;\;\; \textrm{ and } \; \; \;   \frac{1}{q}\sum_{k=1}^{p+q} k \hat t_{k} s^{\frac{k-q}{q}}$$
This again implies the desired result.
\end{proof}

\section{Application to p -- q duality of 2D quantum gravity}
\label{physics-section}
In this section we apply the previous results to give a conceptual proof, based on the local Fourier transform, of the duality results of Fukuma-Kawai-Nakayama \cite{FKN} concerning 2D quantum gravity.

For positive co-prime integers $p$ and $q$ there exists the so-called $(p,q)$ model of 2D quantum gravity (more precisely one should talk about $(p,q)$ minimal model coupled to 2D gravity). This theory has a partition function $\textrm{Z}_{p,q}$ that can be expressed in terms of a $\tau$-function of the KP hierarchy.
The crucial fact is  that the partition function $\textrm{Z}_{p,q}$ is the square of a function $\tau_{p,q}$ that is the $\tau$-function of the KP hierarchy satisfying certain Virasoro constraints. 

The $(p,q)$ models for varying $p$ and $q$ are not unrelated: There exists a certain duality between the $(p,q)$ and the $(q,p)$ theory. This so-called p -- q duality can be expressed as a relation between the two $\tau$-functions $\tau_{p,q}$ and $\tau_{q,p}$. One of the nice consequences of the duality is that it allows to describe a  theory by a simpler one: The $(p,q)$ model depends non-trivially only on the KP times $t_{i}$ with
$$1 \le i \le p + q \; \; \textrm{ and } \; \; \;  i \not \in \ZZ p$$
Therefore, the p -- q duality can simplify the study of the $(p,q)$ models. For example, the $(3,2)$ model depends a priori on four time variables but via the 3 -- 2 duality it can be expressed in terms of three parameters.

In \cite{LUU} it was shown that the p -- q duality can be expressed in terms of the local Fourier duality of certain connections. We now use the results of the previous section to give a more conceptual proof of this fact and furthermore we will give a Fourier theoretic proof of the results of Fukuma-Kawai-Nakayama   concerning the dynamics of the duality.

It is known that $(p,q)$ theory corresponds to a family of solutions to the string equation 
$$[P(t_{1},\cdots, t_{p+q}),Q(t_{1},\cdots, t_{p+q})]=1 \; \; \; \textrm{ where } \;\;\;
P(t_{1},\cdots, t_{p+q})=\partial_{x}^{p}+ a_{p-2} \partial _{x}^{p-2}+\cdots + a_{0}$$ 
is a differential operator of order $p$ and  $Q$ is a differential operator of order $q$.  
 The method  that allows us to describe the corresponding points   of the Sato Grassmannian  was given in \cite {KS}. For more details see \cite {SCH2}, \cite {KM} or \cite {FKN}.  In this description, the $\tau$-function $\tau_{p,q}$ of the theory is known to satisfy
$$\frac{\partial^{2}}{\partial t_{1}^{2}} \ln \tau_{p,q}=\frac{a_{p-2}}{p}$$ 
It is shown in \cite{FKN} that one simply has
$$\frac{\partial^{2}}{\partial t_{1}^{2}} \ln \tau_{q,p}= \frac{\partial^{2}}{\partial t_{1}^{2}} \ln \tau_{p,q} + C(t_{1},\cdots,t_{p+q})$$
where the correction term is given by
$$C(t_{1},\cdots,t_{p+q}) = \frac{1-q}{2} \cdot \frac{p+q-1}{q(p+q)} \cdot \left (\frac{t_{p+q-1}}{t_{p+q}} \right )^{2}+\frac{p+q-2}{q(p+q)} \cdot \frac{t_{p+q-2}}{t_{p+q}}$$
Note in particular the case where $p$ and $q$ are such that one of the following holds: 
\begin{enumerate}[(i)]
\item
$q \ne 1$ and $q \equiv 1 \mod p$
\item
$q \equiv 2 \mod p$
\end{enumerate}
In these cases the function $\frac{\partial^{2}}{\partial t_{1}^{2}} \ln \tau_{p,q}$ does not depend on one of the variables $t_{p+q-1}$ or $t_{p+q-2}$ and one can specialize the value of this time variable in such a manner that the correction term vanishes:
$$C(t_{1},\cdots,t_{p+q}) \equiv 0$$
To really relate the second derivatives of the $\tau$-functions it is crucial to obtain a relation between the two sets of KP times. This  time dynamics of the duality was obtained by Fukuma-Kawai-Nakayama in \cite{FKN}. We now give a conceptual proof of the relation between the $(p,q)$ times and the $(q,p)$ times based on properties of the local Fourier transform. 

\begin{thm}[Fukuma-Kawai-Nakayama \cite{FKN}]
\label{time-relation-theorem}
Let $p$ and $q$ be positive co-prime integers. Define
$$a_{p+q-k}=\frac{k t_{k}}{(p+q)t_{p+q}} \; \; \textrm{ and } \; \; \; \hat a_{p+q-k}=\frac{k \hat t_{k}}{(p+q)\hat t_{p+q}}$$
Then
$$\hat t_{p+q}= -\frac{q}{p}\cdot t_{p+q}$$
and there are values $a_{n}, \hat a_{n}$ extending the above definition to all $n\ge 1$ such that following two functions are compositional inverses:
$$g_{1}=z\left(1+\sum_{n=1}^{\infty} a_{n} z^{-n} \right)^{1/q} $$
and
$$g_{2}=z\left(1+\sum_{n=1}^{\infty}\hat a_{n} z^{-n} \right)^{1/p}$$
\end{thm}
\begin{proof}
We will be able to deduce this result from Theorem \ref{main-theorem}. 
It is  clear that by Theorem \ref{main-theorem} the times $t_{1},\cdots, t_{p+q}$ and $\hat t_{1},\cdots, \hat t_{p+q}$ can be related via the local Fourier transform (see (\ref{mt})).
Let $f_{1}$ and $f_{2}$ be as in the statement of Theorem \ref{main-theorem}. Namely
$$f_{1}=\frac{1}{p} \cdot \sum_{k=1}^{p+q} k t_{k}u^{\frac{k-p}{p}}$$
and
$$f_{2}=-\frac{1}{q} \cdot \sum_{k=1}^{p+q} k \hat t_{k}u^{\frac{k-q}{q}}$$
Setting $\hat  a_{n}=0$ for all $n \ge p+q$ one sees
$$g_{2}(u)=\left( f_{2}(u^{q}) \right)^{1/p}$$ 
Then one sees that
$$g_{2}(u)^{-1}= (f_{2}^{-1}(u^{p}) )^{1/q}=\left( f_{1}(u^{p}) +\sum_{n=p}^{\infty} b_{n}u^{-n} \right)^{1/q}$$ 
for some values $b_{n}$. Hence 
$$g_{2}(u)^{-1}= u \left ( 1 + \sum_{n=1}^{p+q-1} a_{n}z^{-n} +\sum_{n=p}^{\infty} b_{n}u^{-n-q} \right)^{1/q}$$ 
Therefore one can define $g_{1}$ by $g_{1} : = g_{2}^{-1}$ and one obtains the desired result. 
\end{proof}
Via standard techniques involving formulas for inverse functions one can make the time relation even more explicit. This has been carried out by Fukuma-Kawai-Nakayama in \cite{FKN} and the general formula is given by
$$a_{n}=-\frac{q}{p} \cdot \sum_{k \ge 1} \frac{1}{k} \begin{pmatrix}
(n-p-q)/p\\ k-1 \end{pmatrix} \sum_{m_{1},\cdots,m_{k} \ge 1, \sum m_{i}=n} \hat a_{m_{1}} \cdots \hat a_{m_{k}}$$
Consider for example the $2 - 5$ duality:

Choosing $q=2$ one can let the correction term vanish. The time variables $t_{1}, \cdots,t_{7}$ of the $(5,2)$ model can be described in terms of the time variables $\hat t_{1}, \cdots , \hat t_{7}$ of the $(2,5)$ model. We set
$$t_{7}=\frac{5}{7}$$ 
and hence
$$\hat t_{7}=-\frac{2}{7}$$
Then one can calculate the formulas for the 2 -- 5 duality that were presented in the introduction.

Such dynamics were previously known only through explicit calculations with pseudodifferential operators and hence it was difficult to understand their underlying meaning. Due to the present work, the formulas can be understood conceptually as the dynamics underlying the local Fourier transform of connections on the formal punctured disc.

\textbf{Acknowledgements:}

It is a pleasure to thank  V. Vologodsky, A. Graham-Squire, M. Bergvelt for useful exchanges and the referee for helpful remarks and corrections.

\Addresses

\end{document}